\documentclass[letterpaper, 10pt, conference]{IEEEtran}
\usepackage{graphics} 
\usepackage{epsfig} 
\usepackage{times} 
\usepackage{amsmath} 
\usepackage{amssymb}  
\usepackage{amsmath} 
\usepackage{amssymb}  

\usepackage{color}
\usepackage{hyperref}

\newtheorem{thm}{Theorem}

\newtheorem{defn}{Definition}

\newtheorem{prop}{Proposition}
\usepackage{algorithm, algorithmicx, algpseudocode}
\def\QED{\mbox{\rule[0pt]{1.3ex}{1.3ex}}} 
\newenvironment{proof}{\quad {\it Proof.\,}}{\hfill \QED \par}
\newenvironment{proof-of}[1]{{\it Proof of #1:\,}}{\hfill\QED\par}
\newtheorem{rem}{Remark}

\newcommand{\R}{{\mathbb R}}
\renewcommand{\S}{{\mathbb S}}

\newcommand{\Rnn}{{\mathbb R}_{\ge 0}}
\newcommand{\Rp}{{\mathbb R}_{> 0}}

\newcommand{\cB}{{\mathcal B}}

\newcommand{\Hinf}{{\mathbb{H}_\infty}}

\newcommand{\cH}{{\mathcal H}}

\newcommand{\cJ}{{\mathcal J}}

\newcommand{\cM}{{\mathcal M}}

\newcommand{\cV}{{\mathcal V}}

\newcommand{\bfone}{{\boldsymbol 1}}

\newcommand{\diag}[1]{\textrm{diag}\left\{#1\right\}}

\renewcommand{\Re}{\mathrm{Re}}

\newcommand{\blkdiag}[1]{\textrm{blkdiag}\{#1\}}

\begin{document}
\title{Block-Diagonal Solutions to Lyapunov Inequalities and Generalisations of Diagonal Dominance
\author{Aivar Sootla, Yang Zheng and Antonis Papachristodoulou
\thanks{The authors are with the Department of Engineering Science, University of Oxford, Parks Road, Oxford, OX1 3PJ, U.K. e-mail: \{aivar.sootla, yang.zheng, antonis\}@eng.ox.ac.uk. A. Sootla and A. Papachristodoulou are supported by the EPSRC Grant EP/M002454/1. The authors would like to thank Dr James Anderson for valuable discussions. }}}

\maketitle 
\begin{abstract} Diagonally dominant matrices have many applications in systems and control theory. Linear dynamical systems with scaled diagonally dominant drift matrices, which include stable positive systems, allow for scalable stability analysis. For example, it is known that Lyapunov inequalities for this class of systems admit diagonal solutions. In this paper, we present an extension of scaled diagonally dominance to block partitioned matrices. We show that our definition describes matrices admitting block-diagonal solutions to Lyapunov inequalities and that these solutions can be computed using linear algebraic tools. We also show how in some cases the Lyapunov inequalities can be decoupled into a set of lower dimensional linear matrix inequalities, thus leading to improved scalability. We conclude by illustrating some advantages and limitations of our results with numerical examples.
\end{abstract}

\section{Introduction}
Systems admitting diagonal matrix solutions to Lyapunov inequalities are of particular interest in control theory, since they allow for a lower computational complexity of stability analysis. Necessary and sufficient conditions for the existence of diagonal solutions were derived in~\cite{carlson1992block}, which are, however, hard to check. On the other hand, it is well-known that stable linear systems invariant on the positive orthant (or positive systems) admit diagonal matrix solutions to Lyapunov inequalities~\cite{berman1994nonnegative}. Therefore, generalisations of positivity attracted some attention as well, e.g. eventual positivity~\cite{altafini2015realizations, sootla2015evpos}, which inherits some properties of positivity. However, in the context of Lyapunov inequalities, perhaps, a more relevant generalisation is based on (scaled) diagonally dominant matrices. These are defined through constraints on the absolute values of the individual entries of the matrix. Under some conditions scaled diagonally dominant drift matrices admit diagonal solutions to Lyapunov inequalities~\cite{hershkowitz1985lyapunov}, which can be computed using linear programming~\cite{sootla2016existence}.

A block generalisation of diagonal dominance can be obtained by partitioning the matrix into blocks and applying the diagonal dominance constraints to some norms of these blocks as in~\cite{feingold1962block}.
Although some authors considered block versions of scaled diagonal dominance~\cite{sootla2016existence,polman1987incomplete, xiang1998weak}, construction of block-diagonal solutions to Lyapunov inequalities was not fully addressed. In this paper, we present another block generalisation of scaled diagonal dominance. In comparison to previous works, our definition appears to be more suitable for stability analysis, since it includes a control theoretic concept of the $\Hinf$ norm. Our block generalisation of diagonal dominance is consistent with the network dissipativity results in~\cite{cook1974stability}. However, the derivation of block diagonal solutions to Lyapunov inequality was not addressed in~\cite{cook1974stability}.	

It is fairly computationally cheap to check if a matrix satisfies our definition of scaled diagonal dominance facilitating stability analysis of large-scale systems. We show that the introduced class of matrices admits block-diagonal solutions to Lyapunov inequalities, which can be constructed by solving a set of Riccati equations of smaller dimensions. This leads to reduced memory requirements and computational complexity. One can also replace Riccati equations with linear matrix inequalities of smaller dimensions (with respect to the Lyapunov inequality) and optimise over possible solutions. 

The rest of the paper is organised as follows. We cover the preliminaries of positive systems theory, scaled diagonal dominance and some facts from systems theory in Section~\ref{s:prel}. We introduce our extension of scaled diagonal dominance to block partitioned matrices in Subsection~\ref{ss:def}. 
We show how block-diagonal solutions to Lyapunov inequalities are constructed in Subsection~\ref{ss:sa}, and present decoupled stability tests based on our results in Subsection~\ref{ss:test}. We illustrate our results on numerical examples in Section~\ref{s:app} and conclude in Section~\ref{s:con}. The proofs of some auxiliary results are found in the Appendix.

\emph{Notation.} Let $\S^k_+$ (respectively, $\S^k_{++}$) denote the set of $k\times k$ positive semidefinite (respectively, positive definite) matrices in $\R^{k\times k}$. We also write $A\succeq 0$ if $A\in \S^k_+$, and $A\succ 0$ if $A\in \S^k_{++}$. We denote the positive orthant $\Rp^n$, that is, the set of all vectors $x$ with positive entries.
The operator $\cdot^\ast$ denotes a matrix transpose. We denote the maximal singular value of a matrix $A$ as $\overline{\sigma}(A)$, while the minimal as $\underline{\sigma}(A)$. The $\Hinf$ norm of an asymptotically stable transfer function $G(s)$ is computed as $\|G\|_\Hinf =\max_{w\in \R}\|G(\imath \omega)\|_2$, where $\imath$ is the imaginary unit and $\|A\|_2$ is the induced matrix norm equal to $\overline{\sigma}(A)$. Finally, let $\diag{A_1, \dots, A_n}$ denote a block-diagonal matrix with matrices $A_i$ on the block-diagonal.

\section{Preliminaries and Problem Formulation~\label{s:prel}}
Consider the linear time invariant system 
\begin{gather}
\label{eq:sys}
\begin{aligned}
\dot{x}(t)&=A x(t) + B u(t),\\ 
y(t) &= C x(t),
\end{aligned}
\end{gather}
with the transfer function $G(s) = C (sI - A )^{-1} B$, where $x(0)=x_0$, $x(t)\in \R^n$, $u(t) \in \R^m$, $y(t)\in \R^k$. It can be shown~\cite{ZDG} that system~\eqref{eq:sys} is stable with $u(t) = 0$ for all $t$ if and only if there exists $P\succ 0$ such that
\begin{equation}\label{eq:lyap_lmi}
P A + A^\ast P\prec 0.
\end{equation}
The linear matrix inequality (LMI)~\eqref{eq:lyap_lmi} is called a Lyapunov inequality, and its solution defines a Lyapunov function of the form $V(x)=x(t)^\ast P x(t)$ for  system~\eqref{eq:sys} with $u(t) = 0$.  
We will also use the following result from control theory called the Bounded Real Lemma~\cite{ZDG}.
\begin{prop}\label{prop:brl}
	For system~\eqref{eq:sys}, there exists~$\gamma$ such that $\gamma > \|G\|_\Hinf$ if and only if there exists $P\succ 0$ such that
	\begin{gather}
	P A + A^\ast P + P B B^\ast P + C^\ast C \gamma^{-2} \prec 0.
	\end{gather}
\end{prop}

In some cases, we can guarantee the existence of a diagonal matrix $P$ satisfying~\eqref{eq:lyap_lmi}. One of such cases is the class of dynamical systems with \emph{Metzler} drift matrices (or positive systems).

\begin{defn}
	A matrix $A\in \R^{n\times n}$ is said to be  \emph{Metzler} if all the off-diagonal elements are nonnegative.
\end{defn}

Analysis of positive systems is computationally and conceptually simpler than analysis of general types of systems. For example, the following result (which is a combination of results in~\cite{fan1958topological, varga1976recurring, rantzer2015ejc}), allows one to replace semidefinite constraints in analysis and design methods with linear ones, which leads to scalable algorithms. 

\begin{prop}\label{prop:pos-stab}
	Consider a system $\dot x = A x$ with a Metzler matrix $A$. Then the following statements are equivalent:
	
	i) There exists $d \in \Rp^n$ such that $-A d \in \Rp^n$;  
	
	ii) There exists $e\in \Rp^n$ such that $-e^\ast A  \in \Rp^n$; 
	
	iii) $A$ is Hurwitz (has eigenvalues with negative real parts). 
	
	iv) There exists a diagonal $P$ such that $P A + A^\ast P \prec 0$.
\end{prop}

The points (i) and (ii) in Proposition~\ref{prop:pos-stab} imply that Hurwitz Metzler matrices belong to another well-known class of matrices. 

\begin{defn}
	A matrix $A \in\R^{n\times n}$ with entries $a_{i j}$ is called \emph{strictly row scaled diagonally dominant} if there exist positive scalars $d_1, \dots, d_n$ such that:
	\begin{gather}
	d_i |a_{i i}| > \sum\limits_{j=1,j\ne i}^n d_j |a_{i j}|\,\,\forall i=1,\dots,n. \label{sdd:row}
	\end{gather}
The matrix $A$ is called \emph{strictly column scaled diagonally dominant} if there exist positive scalars $e_1, \dots, e_n$ such that:
 	\begin{gather}
 	e_i |a_{i i}| > \sum\limits_{j=1,j\ne i}^n e_j |a_{j i}|\,\,\forall i=1,\dots,n. \label{sdd:column}
 	\end{gather}
 The matrix $A$ is \emph{strictly row (respectively, column) diagonally dominant} if $d_i = 1$ (respectively, $e_i = 1$) for all $i$.  	
\end{defn}

In order to illustrate the connection between diagonal dominance and positivity we introduce the following concept. 
\begin{defn}\label{def:h-mat}
The matrix $\cM(A)$ is called the comparison matrix if its entries $\cM_{ij}(A)$ are defined as
\begin{gather}
\cM_{ij}(A) = \left\{\begin{array}{ll} -\max\{-a_{i i}, 0\} &  \text{if }i = j, \\
|a_{i j}| & \textrm{otherwise.}
\end{array} \right.\label{block-comparison-scalar}
\end{gather}
\end{defn}	

We slightly modified the definition of the comparison matrix compared to the classic one (cf.~\cite{hershkowitz1985lyapunov}) in order to streamline the stability analysis. 
For example, if $A$ is Metzler or $A$ is lower triangular ($a_{i j} =0$, for all $i< j$) then $A$ is Hurwitz if and only if $\cM(A)$ is Hurwitz.
More generally, if $\cM(A)$ is Hurwitz, then $A$ is Hurwitz and $A$ admits a diagonal solution to~\eqref{eq:lyap_lmi}~\cite{hershkowitz1985lyapunov}, which can be constructed using linear algebra, linear or second order cone programming~\cite{sootla2016existence}. In the proof of these results Proposition~\ref{prop:pos-stab} is applied to a Hurwitz Metzler matrix $\cM(A)$, which leads to existence of positive $d_i$, $e_i$ such that~\eqref{sdd:row} and~\eqref{sdd:column} hold, that is $A$ is strictly row and column scaled diagonally dominant. We, finally, note that the matrices with Hurwitz $\cM(A)$ belong to a well-studied class of matrices called \emph{$\cH$ matrices}. We will not discuss in detail this class of matrices, but refer the reader to~\cite{hershkowitz1985lyapunov}, \cite{polman1987incomplete} for details.

In this paper, we discuss a generalisation of scaled diagonal dominance to block partitioned matrices. We say that a matrix $A\in\R^{N\times N }$ has \emph{$\alpha=\{k_1, \dots, k_n\}$-partition} with $N = \sum\limits_{i = 1}^n k_i$, if the matrix $A$ is written as follows
\[
A = \begin{pmatrix}
A_{1 1}    & A_{1 2}     & \dots   & A_{1 n} \\
A_{2 1}    & A_{2 2}     & \dots   & A_{2 n} \\
\vdots     & \vdots      & \ddots  & \vdots  \\
A_{n 1}    & A_{n 2}     & \dots   & A_{n n} 
\end{pmatrix},
\]
where $A_{i j}\in\R^{k_i\times k_j}$. We say that $A$ is \emph{$\alpha$-diagonal} if it is $\alpha$-partitioned and $A_{i j} = 0$ if $i\ne j$.  
We aim at characterising \emph{$\alpha$-diagonally stable} matrices $A\in\R^{N\times N}$ such that there exists an $\alpha$-diagonal positive definite $X\in\R^{N\times N}$ satisfying~\eqref{eq:lyap_lmi}. If the partition is trivial, i.e., $\alpha = \{1,\dots,1\}=\bfone$, we will not mention $\alpha$ and say that an $\alpha$-diagonal (respectively, $\alpha$-diagonally stable) matrix $A$ is \emph{diagonal} (respectively, diagonally stable). We will also use a version of the \emph{Gershgorin circle} theorem for the $\alpha$-partitioned matrices.
\begin{prop}[\cite{feingold1962block}]\label{prop:block-gershgorin}
	For an $\alpha$-partitioned matrix $A\in\R^{N\times N}$, where $\alpha = \{k_1, \dots, k_n\}$ and $N = \sum_{i = 1}^n k_i$, every eigenvalue of $A$ satisfies 
	\begin{gather*}
	\|(\lambda I - A_{ i i})^{-1}\|_2^{-1} \le \sum\limits_{j = 1, j\ne i}^n \|A_{i j}\|_2
	\end{gather*}
	for at least one $i$ where $i = 1, \dots, n$. 
\end{prop}

\section{Generalisations of Diagonal Dominance} \label{s:sdd}
\subsection{$\alpha$-Comparison Matrix and its Properties}\label{ss:def}
We start by introducing a novel generalisation of the comparison matrix $\cM(A)$ to the block partitioned case.
\begin{defn}\label{def:block-comp-1}
	Given an $\alpha$-partitioned matrix $A$, we define the matrix $\cM^\alpha(A)$  as follows
	\begin{gather}
	\cM^\alpha_{ij}(A) = \left\{\begin{array}{ll} -\|(sI - A_{ i i})^{-1}\|_\Hinf^{-1} &  \text{if }i = j, \\
	\|A_{i j}\|_2 & \textrm{otherwise}.
	\end{array} \right.\label{block-comparison-1}
	\end{gather}	
\end{defn}

If $A_{i i}$ is not Hurwitz, then we can continuously extend the function $\|(sI  - A_{ i i})^{-1}\|_\Hinf^{-1}$ so that $\|(sI  - A_{ i i})^{-1}\|_\Hinf^{-1} = 0$. Therefore, Definition~\ref{def:block-comp-1} is well-posed. In~\cite{feingold1962block}, \cite{xiang1998weak}, a similar definition of $\cM^\alpha(A)$ is used, but $\|(sI - A_{ i i})^{-1}\|_\Hinf^{-1}$ is replaced by $\|A_{ i i}^{-1}\|_2^{-1}$ and for stability analysis it is required that $A_{i i}$ are Metzler and Hurwitz. Since $\|(sI - A_{ i i})^{-1}\|_\Hinf = \|A_{i i}^{-1}\|_2$ for Hurwitz Metzler matrices and we do not have any restrictions on $A_{i i}$ besides stability, our definition appears to be better suited for stability analysis. It is tempting to call the set of matrices such that the comparison matrices  $\cM^\alpha(A)$ are Hurwitz as block scaled diagonally dominant or block-$\cH$ matrix similarly to~\cite{feingold1962block}, \cite{xiang1998weak}. There are, however, several definitions of block-$\cH$ matrices and in order to minimise confusion we will resist of introducing new nomenclature.  

Now we will discuss the properties of our extension. If $\cM^\alpha(A)$ is Hurwitz then according to Proposition~\ref{prop:pos-stab} there exist positive scalars $d_i$, $e_i$ such that for all $i= 1,\dots, n$:
\begin{gather}
\|(sI - A_{i i})^{-1}\|_\Hinf^{-1} d_i > \sum\limits_{j=1,j\ne i}^n \overline{\sigma}(A_{i j}) d_j, \label{sdd-cond-1}\\ 
\|(sI - A_{i i})^{-1}\|_\Hinf^{-1} e_i  > \sum\limits_{j=1,j\ne i}^n \overline{\sigma}(A_{j i}) e_j.\label{sdd-cond-2} 
\end{gather}  
In the trivial partition case, i.e., $\alpha = \{1,~\dots,~1\}=\bfone$, we have $\|(sI  - a_{ i i})^{-1}\|_\Hinf^{-1} = \max\{-a_{i i}, 0\}$ and $a_{i i}$  is Hurwitz if and only if it is negative. Therefore, the $\alpha$-partitioned generalisation of the comparison matrix reduces to our previous definition. Furthermore, 
stability of the matrix $\cM^\bfone(A)$ ensures stability of the matrix $A$. In the $\alpha$-partitioned case, a similar statement can be made. 

\begin{prop} \label{prop:stab-comp}
An $\alpha$-partitioned matrix $A$ is Hurwitz, if $\cM^\alpha(A)$ is Hurwitz.
\end{prop}

The proof can be found in the Appendix. In what follows we will show that stability of $\cM^\alpha(A)$ implies a stronger property of $A$, namely, $\alpha$-diagonal stability, a result which carries over from the case $\alpha = \bfone$. However, some of the properties of scaled diagonally dominant matrices are not preserved in our generalisation. 
\begin{prop}\label{prop:h-counter}
	There exists a matrix $A$ and a partition $\alpha$ such that $\cM^\alpha(A)$ is a Hurwitz matrix, however, for any $\alpha$-diagonal positive definite $P$ the matrix $\cM^\alpha(A^\ast P + P A)$ is not a Hurwitz matrix.
\end{prop}

The proof can be found in Appendix. This proposition seems to add just a minor detail, however, for Hurwitz $\cM^\bfone(A)$ and any diagonal solution to its Lyapunov inequality $P$ we have that $\cM^\bfone(P A + A^\ast P)$ is Hurwitz. This property was used to construct a diagonal $P$, therefore in what follows we need to find another technique for the $\alpha$-partitioned case.

\subsection{Computation of $\alpha$-diagonal Lyapunov Matrices} \label{ss:sa}  
We start by considering the following auxiliary result.

\begin{prop} \label{prop:comp-mat-lmi}
	Let $\cM^\alpha(A)$ be a Hurwitz matrix, then there exist $\gamma_{i j}\in \Rnn^n$, $W_{i j} \in \S_+^{k_i}$, $P_i \in \S_{+ +}^{k_i}$ such that
	\begin{gather}
	\begin{gathered}
	P_i A_{i i}  + A_{i i}^\ast P_i + \gamma_{i i} I_{k_i} + W_{i i} \preceq 0, \\
	\begin{pmatrix}
	W_{i j} & -P_i A_{i j} \\
	-A_{i j}^\ast P_i & \gamma_{i j} I_{k_j}
	\end{pmatrix} \succeq 0, \\
	\gamma_{i i} > \sum\limits_{j=1,j\ne i}^n \gamma_{j i},\,\,W_{i i}\succ \sum\limits_{j=1,j\ne i}^n W_{i j}.
	\end{gathered}\label{eq:scalable-lmi}
	\end{gather}
\end{prop}
\begin{proof}
	The proof is constructive and we find explicitly $P_i$, $W_{i j}$ and $\gamma_{i j}$ satisfying LMIs~\eqref{eq:scalable-lmi}. Let $A$ be $\alpha$-partitioned and let $\cM^\alpha(A)$ be a Hurwitz matrix, which implies that there exist positive scalars $d_i$, $e_i$ such that~\eqref{sdd-cond-1},~\eqref{sdd-cond-2} hold. Let also 
	\begin{multline} 
	A_{i, -i} =\begin{pmatrix}
	A_{i, 1} & \cdots & A_{i, i-1} & A_{i, i+1} & \cdots & A_{i, n}
	\end{pmatrix}\\ 
	\Gamma_i = \diag{\begin{pmatrix}
		\gamma_{i, 1} I& \cdots &\gamma_{i, i-1} I &\gamma_{i, i+1} I & \cdots &\gamma_{i, n} I
		\end{pmatrix}}\\
	\text{where } \gamma_{i j} = \begin{cases} 
	\|(sI - A_{i i})^{-1}\|_\Hinf^{-1} e_i / d_i & \text{ if } i= j \\
	\overline\sigma (A_{i j}) e_i /d_j   & \text{ otherwise}\\
	\end{cases}. \label{mat-def}
	\end{multline}
	The scalars $\gamma_{i j}$ are equal to zero, if $A_{i j}$ is. Therefore, we introduce matrices $\tilde \Gamma_i$ and $\tilde A_{i,-i}$, which are obtained by removing all zero blocks from $\Gamma_i$ and $A_{i,-i}$. Let $\cJ_i = \{j \in [1,\dots, n] | \gamma_{i j} \ne 0, j\ne i \}$. We have that
	\begin{multline}
	\|\tilde A_{i, -i} \tilde\Gamma_i^{-1/2}\|_2^2 =  \overline{\sigma}\left(\sum_{j\in \cJ_i} A_{i j} A_{i j}^\ast/\gamma_{i j} \right) \le \\
	\sum_{j \in \cJ_i} \overline{\sigma}\left( A_{i j} A_{i j}^\ast \right)/\gamma_{i j} = 
	\sum_{j \in \cJ_i} (\overline{\sigma}(A_{i j}))^2/\gamma_{i j} = \\
	\sum_{j=1,j\ne i}^n \overline{\sigma}(A_{i j}) d_j/e_i <
	\|(sI - A_{i i})^{-1}\|_\Hinf^{-1} d_i /e_i.\label{ineq:main}
	\end{multline}
	
	Using inequality~\eqref{ineq:main} we can obtain further bounds
	\begin{multline*}
	\|(sI - A_{i i})^{-1} \tilde A_{i, -i} \tilde \Gamma_i^{-1/2}\|_\Hinf^2 \le \\
	\|(sI - A_{i i})^{-1}\|_\Hinf^2 \|\tilde A_{i, -i} \tilde \Gamma_i^{-1/2}\|_2^2  <  \\
	\|(sI - A_{i i})^{-1}\|_\Hinf d_i/e_i = \gamma_{i i}^{-1}.
	\end{multline*}
	This according to Proposition~\ref{prop:brl} implies that for some $P_i\succ 0$
	\begin{gather*}
	P_i A_{i i}  +  A_{i i}^\ast P_i + \gamma_{i i} I  + P_i \tilde A_{i,-i} \tilde \Gamma_i^{-1} \tilde A_{i,-i}^\ast P_i \prec 0.
	\end{gather*}
	By noticing that $\gamma_{ii} > \sum_{j=1,j\ne i}^n \gamma_{j i}$  and introducing new variables $W_{i j} \succeq P_i A_{i j} A_{ij}^\ast P_i/\gamma_{i j}$, we obtain~\eqref{eq:scalable-lmi}. 
\end{proof}

There is a certain dimensional asymmetry in the seemingly related variables $W_{i j}$ (which is a matrix), $\gamma_{i j}$ (which is a scalar) in~\eqref{eq:scalable-lmi}. Actually, if the main goal in mind is stability analysis, we can relax the conditions~\eqref{eq:scalable-lmi} and consider the following LMIs with $P_i \in \S_{++}^{k_i}$, $W_{i j}, V_{j i} \in \S_+^{k_i}$.
\begin{subequations}
\begin{gather}
P_i A_{i i}  + A_{i i}^\ast P_i + V_{i i} + W_{i i} \preceq 0, \label{eq:gen-scalable-lmia}
\\
\begin{pmatrix}
W_{i j} & -P_i A_{i j} \\
-A_{i j}^\ast P_i & V_{i j}
\end{pmatrix} \succeq 0, \label{eq:gen-scalable-lmib}
\\
V_{i i} \succ \sum_{j=1,j\ne i}^n V_{j i},\,\,W_{i i}\succ \sum_{j=1,j\ne i}^n W_{i j},\label{eq:gen-scalable-lmic}
\end{gather}
\label{eq:gen-scalable-lmi}
\end{subequations}
which leads to the main theoretical result of the paper.

\begin{thm} \label{lem:lmi}
	Let an $\alpha$-partitioned $A$ satisfy~\eqref{eq:gen-scalable-lmi} for some $P_i \in \S_{++}^{k_i}$, $W_{i j}, V_{j i} \in \S_+^{k_i}$, then  the matrix $A$ is $\alpha$-diagonally stable. Furthermore, $P A + A^\ast P \prec 0$ with $P = \diag{P_1,\dots,P_n}\succ 0$. 
\end{thm}
\begin{proof}
	Let $R_{i j} = \begin{pmatrix}
	R_{i i} &
	R_{j j}
	\end{pmatrix}$ for $i\ne j$, where $R_{i i}\in \R^{N \times k_i}$ partitioned into blocks of the size $k_i \times k_j$ for all $j = 1,\dots,n$. All the blocks are zero matrices, except for the $i$'s block entry, which is an identity matrix. We have the following decomposition:
	\begin{multline*}
	P A + A^\ast P = \sum\limits_{i = 1}^n \Bigl(R_{i i}(P_i A_{i i}  +  A_{i i}^\ast P_i + V_{i i}  + W_{i i}) R_{i i}^\ast - \\
	- R_{i i}(V_{i i} + W_{i i} -\sum\limits_{j=1,j\ne i}^n (W_{i j} +V_{j i}) )R_{i i}^\ast\Bigl) + \\
	\sum\limits_{i = 1}^n \sum\limits_{j=1,j\ne i}^n R_{i j}\begin{pmatrix}
	-W_{i j} & P_i A_{i j} \\
	A_{i j}^\ast P_i & -V_{i j}
	\end{pmatrix}R_{i j}^\ast.
	\end{multline*}
	It is straightforward to show that $\sum_{i = 1}^n R_{i i} (- V_{i i} + \sum_{j=1,j\ne i}^n V_{j i}) R_{i i}^\ast$, $\sum_{i = 1}^n R_{i i} (- W_{i i} + \sum_{j=1,j\ne i}^n W_{i j}) R_{i i}^\ast$ are negative definite, while other sums are negative semidefinite, therefore $P A + A^\ast P \prec 0$. 
\end{proof}	

This result also implies that a matrix $A$ is $\alpha$-diagonally stable provided that $\cM^\alpha(A)$ is Hurwitz. The class of matrices satisfying LMIs~\eqref{eq:gen-scalable-lmi} can be seen as a generalisation of matrices with a stable $\cM(A)$ in their own right. 
Let $\alpha = \bfone$ and $A = \{a_{i j}\}_{i,j =1}^n$, then constraints~\eqref{eq:gen-scalable-lmi} are simplified to
	\begin{gather}
	\begin{gathered}
	-a_{i i} \ge \frac{w_{i i} + v_{i i}}{2 p_i}, \,\, |a_{i j}| \le \frac{\sqrt{w_{i j} v_{i j}}}{p_i}, \\	
	w_{i i} > \sum\limits_{j=1,j\ne i}^n w_{i j}, \,\,  v_{i i} >  \sum\limits_{j=1,j\ne i}^n v_{j i}.
	\end{gathered}\label{eq:lin-cond}
	\end{gather}
\begin{prop}
	The matrix $A$ satisfies~\eqref{eq:lin-cond} if and only if $\cM(A)$ is Hurwitz.
\end{prop}	
\begin{proof}
  	     If $\cM(A)$ is Hurwitz then according to Proposition~\ref{prop:pos-stab} it is diagonally stable and as a result $a_{i i}$ are negative. Furthermore, there exist positive $e_i$, $d_i$ such that $-a_{i i} d_i > \sum_{j=1,j\ne i}^n |a_{i j}| d_j$ and $-a_{i i} e_i > \sum_{j=1,j\ne i}^n |a_{j i}| e_j$. Now we can set $w_{i j} =  |a_{i j}| e_i d_j/d_i^2$, $v_{i j}= |a_{i j}| e_i/d_j$, $p_i = e_i/d_i$ and verify that $A$ satisfies~\eqref{eq:lin-cond}. Now let~\eqref{eq:lin-cond}  be fulfilled. Consider a strictly column diagonally dominant (and hence Hurwitz) Metzler matrix 
		\begin{gather*}
		\cV = \begin{cases}
		- v_{i i} & i = j \\
		v_{i j} & i\ne j
		\end{cases}.
		\end{gather*}
		This implies that there exist positive scalars $d_i$ such that $v_{i i} d_i^2 > \sum_{j=1,j\ne i}^n v_{i j} d_j^2$ according to Proposition~\ref{prop:pos-stab}. For any positive scalars $x$, $y$ we have $\sqrt{x y} \le (x + y)/2$, therefore
		\begin{multline*}
		\sum\limits_{j=1,j\ne i}^n |a_{i j}| d_j/d_i \le \sum\limits_{j=1,j\ne i}^n \frac{\sqrt{v_{i j} (d_j/d_i)^2 w_{i j}}}{p_i} \le \\
		\frac{1}{2 p_i} \sum\limits_{j=1,j\ne i}^n v_{i j} (d_j/d_i)^2 +w_{i j} < \frac{1}{2 p_i} (w_{i i} + v_{i i}) \le -a_{ii},
		\end{multline*}
		which shows that $\cM(A)$ is strictly row scaled diagonally dominant and hence Hurwitz.    
\end{proof}
\subsection{Decoupled Stability Tests}\label{ss:test}
According to Theorem~\ref{lem:lmi}, we can test stability of an $\alpha = \{k_1, \dots, k_n\}$-partitioned matrix $A$ using LMIs~\eqref{eq:gen-scalable-lmi}. These LMIs provide only a sufficient condition for stability, but they are decentralised in the sense that the semidefinite constraints are of orders $k_i$, and we do not need to impose a semidefinite constraint of order $\sum_{i =1}^n k_i$. We can fully decouple the stability tests by setting, for example, $V_{i j} = \gamma_{i j} I_{k_j}$ for some fixed $\gamma_{i j}$, while eliminating $W_{i j}$ using the Schur complement formula. We  get
\begin{multline}
P_i A_{i i}  + A_{i i}^\ast P_i + P_{i} \left(\sum\limits_{j=1,j\ne i}^n A_{i j} A_{i j}^\ast/\gamma_{i j}\right) P_i + \\
I_{k_i} (\varepsilon_i + \sum\limits_{j=1,j\ne i}^n \gamma_{j i}) = 0, \forall i =1,\dots, n.
\label{eq:decoupled-lmi}
\end{multline}
where $\varepsilon_i$ are positive predefined scalars. The choice of the gains $\gamma_{i j}$ is essential and we present a few ad-hoc choices. 

{\bf Test A.} The equations~\eqref{eq:decoupled-lmi} have solutions $P_i \succ 0$ with $\gamma_{i j} = \overline{\sigma}(A_{i j})$.

{\bf Test B.} The equations~\eqref{eq:decoupled-lmi} have solutions $P_i \succ 0$ with
\begin{gather*}
\gamma_{i j} = \begin{cases}
1 & \overline\sigma(A_{i j}) > 0, \\
0 & \overline\sigma(A_{i j}) = 0.
\end{cases}
\end{gather*}

{\bf Test C.} The matrix $\cM^\alpha(A)$ is Hurwitz, that is there exist positive scalars $e_i$, $d_i$ such that~\eqref{sdd-cond-1},~\eqref{sdd-cond-2} hold, which implies that the equations~\eqref{eq:decoupled-lmi} have solutions $P_i \succ 0$ with $\gamma_{i j} = \overline{\sigma}(A_{i j}) e_i/d_j$.

The sets of matrices satisfying these stability tests intersect, but none of them includes the other. It is possible to find matrices, which satisfy only one of tests and fail two others. This cannot be done in the trivial partition case (i.e., $\alpha = \{1,\dots, 1\}$), but in the two block case with $\alpha = \{2, 2\}$. we present such examples in what follows.

\section{Numerical Examples}\label{s:app}
\begin{figure*}
	\small 	\begin{equation}
	A = \left(\begin{array}{cc|cc}
	-67&   -30 &  2 & 8\\
	20 &  -27  &  2 &5\\
	\hline
	14 &  -10  & -57  &  40\\
	-3 &   10  &  50 &-27
	\end{array}\right)\,
	B = \left(\begin{array}{cc|cc}
	-30  &  30   &  0   & 2\\
	50   & -61 & -6 &-8\\
	\hline
	3 & -10   & -53 &-40 \\
	13 &  13 &  10   &-73
	\end{array} \right)\,C = \left(\begin{array}{cc|cc}
	-60  &  30  &  6  &  6    \\
	20   & -20   &      0 &   7 \\
	\hline
	7  &  2 &  -90   & 20   \\
	7  & -5 &    0  & -20  
	\end{array}\right) \label{examples}
	\end{equation} 
	\vspace{-20pt}
\end{figure*}
\emph{Example 1.} First, we present an example verifying the result of Proposition~\ref{prop:h-counter}. Let $\delta = 1.63$, $Q = \blkdiag{Q_1, Q_1}$
\begin{gather*}
A = \begin{pmatrix}
B & \delta I \\
\delta I & B 
\end{pmatrix},\,\, B = \begin{pmatrix}
-8 & 8\\
5 & -8
\end{pmatrix},\,\,Q_1 = \begin{pmatrix}
7 & 7\\ 7 & 11 \end{pmatrix}.
\end{gather*}
With $\alpha = \{2,2\}$, the $\cM^{\alpha}(A)$ is Hurwitz, and the matrix $Q A +  A^\ast Q$ is negative definite. We can verify if there exists $P = \blkdiag{P_1, P_2} \succ 0$ such that $\cM^\alpha(P A + A^\ast P)$ is a Hurwitz matrix using LMIs. In particular, it can be shown that $\|(sI - P A - A^\ast P)^{-1}\|_\Hinf^{-1} = \underline{\sigma}(P A + A^\ast P)$, hence we have the following matrix inequalities:
\begin{gather*}
P_{i} A_{i i} + A_{i i}^\ast P_i \preceq -\gamma_{ii} I_{k_i}, i = 1,2\\
\begin{pmatrix}
\gamma_{1 2} I_{k_1} & P_{1} A_{1 2} + A_{2 1}^\ast P_2 \\
P_{2} A_{2 1} + A_{1 2}^\ast P_1 & \gamma_{1 2} I_{k_2}
\end{pmatrix} \succeq 0, \\
\cB =\begin{pmatrix}
-\gamma_{1 1} & \gamma_{1 2}\\
\gamma_{1 2} & -\gamma_{2 2}
\end{pmatrix} \prec 0,
\end{gather*}
where $\cB \prec 0$ is equivalent to $\cB$ being Hurwitz for symmetric matrix $\cB$. Numerical computations show that there exists no $P = \blkdiag{P_1, P_2} \succ 0$ such that $\cM^\alpha(P A + A^\ast P)$ is Hurwitz. But if we set $\delta = 1.6$, then it is straightforward to check that $\cM^\alpha(Q A + A^\ast Q)$ is Hurwitz.

\emph{Example 2. }  Consider matrices $A$, $B$ and $C$ in~\eqref{examples}. It can be verified that the matrix $A$ satisfies Test~A and fails Tests~B,C, the matrix $B$ satisfies Test~B and fails Tests~A,C and finally the matrix $C$ satisfies Test~C, while fails Tests~A,B.

Test B uses binary information about the interconnections (if they exist or not), while Test~A uses also the information on the gains of the interconnections. Therefore, it may seem counter-intuitive that Test A sometimes fails when Test B prevails, since in Test A we use more information about the system than in Test B. In our examples, if the gain $\overline{\sigma}(A_{1 2})$ is larger than one, then we make the Riccati equation for $i =1$ less conservative by normalising $A_{1 2}$. However, we make the other Riccati equation (with $i=2$) more conservative by increasing the term $\gamma_{1 2} I$, which requires to make the $\Hinf$ norm of the system $(s I -A_{2 2})^{-1} A_{2 1}$ smaller. Therefore, for large gains $A_{i j}$ either of the tests can perform better depending on the drift matrices $A_{i i}$. 

The main conservatism of Test C is in the transition to the Riccati equations. We use the Cauchy-Schwartz inequality, which may be conservative if the eigenspaces of the matrices $A_{i i}$, $A_{i j}$ are not aligned. In control-theoretic language, this corresponds to the mode of $A_{i i}$ closest to the imaginary axis being poorly controllable using the input matrix $A_{i j}$. On the other hand, we also scale the gains $\overline{\sigma}(A_{i j})$, which provides extra freedom. 

In our examples, these limitations of the tests are not apparent, which indicates that even the slightest changes in the gains $\overline{\sigma}(A_{i j})$ and the eigenspaces of $A_{i j}$ can result in the failure of one of the completely decoupled tests. 

\emph{Example 3.} We proceed with a rather theoretical observation. It is well-known that an $\alpha$-triangular matrix $A$ is Hurwitz if and only if the blocks on $\alpha$-diagonal are Hurwitz, which also implies that it is $\alpha$-diagonally stable. We will consider this class of matrices through our generalisation of scaled diagonal dominance on a specific example. Let
\begin{gather*}
A = 
\begin{pmatrix}
    -6  &   4  &   0  &   0  &   0  &   0 \\
    8   & -7   &  0   &  0   &  0   &  0 \\
    4   &  6   & -1   & -2   &  4   &  0 \\
    7   & -2   & 3    & -1   &  6   &  0 \\
    1   &  2   &  1   &  0   & -7   &  0 \\
    -1  &   7   &  4  &   6  &  -5  &  -2
\end{pmatrix}.
\end{gather*}
While setting $\alpha = \{2, 3, 1\}$, compute the comparison matrix
\begin{gather*}
\cM^\alpha(A) = \begin{pmatrix}
-0.7799&         0 &        0 \\
8.4427 &  -0.5282  &       0 \\
7.0711 &   8.7750  & -2.0000
\end{pmatrix}, 
\end{gather*}
which is Hurwitz as long as elements on the diagonal are negative. It can be verified that the generalisation of the scaled diagonal dominant matrices from~\cite{feingold1962block,sootla2016existence} will not yield conclusive results on stability analysis of the matrix $A$. Our definition, however, confirms that stability of the blocks $A_{i i}$ (and hence stability of $\cM^\alpha(A)$) is necessary and sufficient for $\alpha$-diagonal stability of an $\alpha$-triangular matrix $A$.

\emph{Example 4. }
Now we consider another class of matrices called border block diagonal~\cite{vsiljak1978large}. Consider the following $\alpha$-partitioned matrix
\begin{gather*}
\small A = \begin{pmatrix}
A_{1 1} & A_{1 2} &  A_{1 3} & \cdots & A_{1 n} \\
A_{2 1} & A_{2 2} &  0       & \cdots & 0       \\
A_{3 1} & 0       &\ddots    & \ddots & \vdots  \\
\vdots  & \vdots  &\ddots    & \ddots & 0       \\
A_{n 1} & 0       &\cdots    & 0      & A_{n n}
\end{pmatrix}.
\end{gather*}
If the matrix $A$ satisfies~\eqref{eq:gen-scalable-lmi} then the conditions are simplified since $A_{i j} = 0$ unless $i =j$, $i =1$ or $j = 1$. We can set directly $V_{i j} =0$, $W_{i j} =0$ if $A_{ij} = 0$, furthermore, we sum LMIs in~\eqref{eq:gen-scalable-lmib} containing $W_{1 i}$ and $W_{i 1}$ for every $i$, after rearranging the LMIs to fit the dimensions. After these operations, it can be shown that conditions~\eqref{eq:gen-scalable-lmi} imply the following LMIs
\begin{gather}
\begin{gathered}
Q_1 A_{1 1} + A_{1 1}^\ast Q_1 +\sum\limits_{j > 1} Y_{j} \prec 0 \\
Q_j A_{j j} + A_{j j}^\ast Q_j + Z_j \preceq 0, j >1 \\
\begin{pmatrix}
Y_{j}                 & -Q_1 A_{1 j} - A_{j 1}^\ast P_j\\
- A_{1 j}^\ast Q_1 -  Q_j A_{j 1}  & Z_j
\end{pmatrix} \succ 0, j >1, 
\end{gathered}\label{con:bbd}
\end{gather}
where we set $Y_j = W_{1 j} + V_{j 1}$, $Z_j = W_{j 1} + V_{1 j}$. These conditions can be obtained directly from the LMI $Q A + A^\ast Q\prec 0$ provided that $Q= \diag{Q_1,\dots, Q_n}\succ 0$ using standard decomposition techniques (cf.~\cite{mason2014chordal},\cite{vandenberghe2015chordal}). Therefore, conditions~\eqref{con:bbd} are necessary and sufficient for $\alpha$-diagonal stability of border block diagonal matrices. Conditions~\eqref{con:bbd} are less restrictive than stability of $\cM^\alpha(A)$ and conditions~\eqref{eq:gen-scalable-lmi} applied to the matrix $A$, at the same time one can view conditions~\eqref{con:bbd} as conditions~\eqref{eq:gen-scalable-lmi} applied to $Q A + A^\ast Q$ with $P_i = I$. 

\section{Conclusion} \label{s:con}
In this paper, we presented a generalisation of scaled diagonal dominance for block partitioned matrices. Our main goal was to provide conditions on the drift matrix, which facilitate the stability analysis of large-scale systems, in the spirit of positive systems theory. In particular, we derived sufficient conditions for existence of block-diagonal solutions to Lyapunov inequalities. We have already noted the similarity of our work to dissipativity theory by pointing out the relation to~\cite{cook1974stability}. In addition to stability results in~\cite{cook1974stability}, we explicitly constructed Lyapunov inequalities and decoupled the stability test into a number of LMIs, which can potentially be used for distributed stability analysis. For example, one can use decomposition techniques similar to~\cite{zheng2016fast}, in order to derive scalable optimisation algorithms. Furthermore, our results can be applied to decentralised control problems as indicated in~\cite{zheng2017blockdiagonal}.
\bibliography{bsdd}

\begin{thebibliography}{10}
\providecommand{\url}[1]{#1}
\csname url@rmstyle\endcsname
\providecommand{\newblock}{\relax}
\providecommand{\bibinfo}[2]{#2}
\providecommand\BIBentrySTDinterwordspacing{\spaceskip=0pt\relax}
\providecommand\BIBentryALTinterwordstretchfactor{4}
\providecommand\BIBentryALTinterwordspacing{\spaceskip=\fontdimen2\font plus
\BIBentryALTinterwordstretchfactor\fontdimen3\font minus
  \fontdimen4\font\relax}
\providecommand\BIBforeignlanguage[2]{{%
\expandafter\ifx\csname l@#1\endcsname\relax
\typeout{** WARNING: IEEEtran.bst: No hyphenation pattern has been}%
\typeout{** loaded for the language `#1'. Using the pattern for}%
\typeout{** the default language instead.}%
\else
\language=\csname l@#1\endcsname
\fi
#2}}

\bibitem{carlson1992block}
D.~Carlson, D.~Hershkowitz, and D.~Shasha, ``Block diagonal semistability
  factors and {L}yapunov semistability of block triangular matrices,''
  \emph{Linear Algebra and its Applications}, vol. 172, pp. 1--25, 1992.

\bibitem{berman1994nonnegative}
A.~Berman and R.~J. Plemmons, \emph{Nonnegative Matrices in the Mathematical
  Sciences}.\hskip 1em plus 0.5em minus 0.4em\relax SIAM, 1994, vol.~9.

\bibitem{altafini2015realizations}
C.~Altafini, ``Representing externally positive systems through minimal
  eventually positive realizations,'' in \emph{Proc IEEE Conf Decision
  Control}, 2015, pp. 3591--3596.

\bibitem{sootla2015evpos}
A.~{Sootla}, ``{Properties of Eventually Positive Linear Input-Output
  Systems},'' \emph{{ArXiv e-prints arXiv:1509.08392}}, Sept 2015.

\bibitem{hershkowitz1985lyapunov}
D.~Hershkowitz and H.~Schneider, ``Lyapunov diagonal semistability of real
  {H}-matrices,'' \emph{Linear algebra and its applications}, vol.~71, pp.
  119--149, 1985.

\bibitem{sootla2016existence}
A.~Sootla and J.~Anderson, ``On existence of solutions to structured {L}yapunov
  inequalities,'' in \emph{Proc of Am Control Conf}.\hskip 1em plus 0.5em minus
  0.4em\relax IEEE, 2016, pp. 7013--7018.

\bibitem{feingold1962block}
D.~G. Feingold, R.~S. Varga, \emph{et~al.}, ``Block diagonally dominant
  matrices and generalizations of the gerschgorin circle theorem,''
  \emph{Pacific J. Math}, vol.~12, no.~4, pp. 1241--1250, 1962.

\bibitem{polman1987incomplete}
B.~Polman, ``Incomplete blockwise factorizations of (block) {H}-matrices,''
  \emph{Linear Algebra Appl}, vol.~90, pp. 119--132, 1987.

\bibitem{xiang1998weak}
S.-h. Xiang and Z.-y. You, ``Weak block diagonally dominant matrices, weak
  block {H}-matrix and their applications,'' \emph{Linear Algebra Appl}, vol.
  282, no.~1, pp. 263--274, 1998.

\bibitem{cook1974stability}
P.~A. Cook, ``On the stability of interconnected systems,'' \emph{International
  Journal of Control}, vol.~20, no.~3, pp. 407--415, 1974.

\bibitem{ZDG}
K.~Zhou, J.~C. Doyle, and K.~Glover, \emph{Robust and optimal control}.\hskip
  1em plus 0.5em minus 0.4em\relax Prentice Hall New Jersey, 1996.

\bibitem{fan1958topological}
K.~Fan, ``Topological proofs for certain theorems on matrices with non-negative
  elements,'' \emph{Monatshefte f{\"u}r Mathematik}, vol.~62, no.~3, pp.
  219--237, 1958.

\bibitem{varga1976recurring}
R.~S. Varga, ``On recurring theorems on diagonal dominance,'' \emph{Linear
  Algebra and its Applications}, vol.~13, no.~1, pp. 1--9, 1976.

\bibitem{rantzer2015ejc}
A.~Rantzer, ``Scalable control of positive systems,'' \emph{European Journal of
  Control}, vol.~24, pp. 72--80, 2015.

\bibitem{vsiljak1978large}
D.~D. {\v{S}}iljak, \emph{Large-scale dynamic systems: stability and
  structure}.\hskip 1em plus 0.5em minus 0.4em\relax North Holland, 1978,
  vol.~2.

\bibitem{mason2014chordal}
R.~P. Mason and A.~Papachristodoulou, ``Chordal sparsity, decomposing sdps and
  the lyapunov equation,'' in \emph{Proc of Am Control Conf}.\hskip 1em plus
  0.5em minus 0.4em\relax IEEE, 2014, pp. 531--537.

\bibitem{vandenberghe2015chordal}
L.~Vandenberghe, M.~S. Andersen, \emph{et~al.}, ``Chordal graphs and
  semidefinite optimization,'' \emph{Foundations and Trends{\textregistered} in
  Optimization}, vol.~1, no.~4, pp. 241--433, 2015.

\bibitem{zheng2016fast}
Y.~Zheng, G.~Fantuzzi, A.~Papachristodoulou, P.~Goulart, and A.~Wynn, ``Fast
  {ADMM} for semidefinite programs with chordal sparsity,'' \emph{arXiv
  preprint arXiv:1609.06068}, 2016.

\bibitem{zheng2017blockdiagonal}
Y.~Zheng, M.~Kamgarpour, A.~Sootla, and A.~Papachristodoulou, ``Convex design
  of structured controllers using block-diagonal {L}yapunov functions.''
  \emph{ArXiv e-prints}, 2017, ar{X}iv:1709.00695.

\end{thebibliography}
\section*{Appendix} 
\begin{proof-of}{Proposition~\ref{prop:stab-comp}}
We note that the proof of the following result employs the technique used in~\cite{feingold1962block}. We prove the result by contradiction. Let $A$ have eigenvalues with a nonnegative real part and let $\cM^\alpha(A)$ be Hurwitz, which implies there exists positive scalars $d_i$ such that~\eqref{sdd-cond-1} holds for every $i$. Since $A$ has eigenvalues with a nonnegative real part, then so does the matrix $D^{-1} A D$ is with $D = \diag{d_1 I_{k_1}, \dots , d_n I_{k_n}}$. Let $\lambda$ be the eigenvalue of $D^{-1} A D$ with a nonnegative real part. By Proposition~\ref{prop:block-gershgorin} there exists an index $i$ such that 
	\begin{gather} \label{block-gershgorin}
	\|(\lambda I - A_{i i})^{-1}\|_2^{-1} \le \sum\limits_{j=1,j\ne i}^n \left\|A_{i j}\frac{d_j}{d_i}\right\|_2 = \sum\limits_{j=1,j\ne i}^n \|A_{i j}\|_2 \frac{d_j}{d_i}.
	\end{gather}
	However, $\|(\lambda I - A_{i i})^{-1}\|_2 \le \|(s I - A_{i i})^{-1}\|_\Hinf$ for any $\lambda$ such that $\Re(\lambda) \ge 0$. Combining~\eqref{sdd-cond-1} and~\eqref{block-gershgorin} gives:
	\begin{multline*}
	\sum\limits_{j=1,j\ne i}^n \|A_{i j}\|_2 \frac{d_j}{d_i} \ge 
	\|(\lambda I - A_{i i})^{-1}\|_2^{-1}\ge \\
	\|(s I - A_{i i})^{-1}\|_\Hinf^{-1} >\sum\limits_{j=1,j\ne i}^n  \|A_{i j}\|_2 \frac{d_j}{d_i}.
	\end{multline*}
	We arrive at the contradiction and complete the proof.
\end{proof-of}
\begin{proof-of}{Proposition~\ref{prop:h-counter}}
	Consider the matrix $A$ with
	\begin{gather*}
	A = \begin{pmatrix}
	B & \delta I \\ \delta I & B
	\end{pmatrix}
	\end{gather*} 
	with $\delta > 0$, Hurwitz matrix $B\in\R^{k\times k}$ such that
	
	(a) $\underline{\sigma}(B) \ge 1$
	
	(b) $\overline{\sigma}(X_0) \delta \ge 1/2$, where $X_0$ is a solution to $X_0 B^\ast + B X_0 + I =0$.
	
	(c) the matrix $\cM^\alpha(A)$ is Hurwitz with $\alpha = \{k ,k\}$.
	
	Such a matrix exists if $B\in\R^{k\times k}$ with $k\ge 2$. For example, let $B = \begin{pmatrix}
	-8 & 8\\
	5 & -8
	\end{pmatrix}$ and $\delta = 1.63$.
	
	We will show that under assumptions (a)-(c) there does not exist an $\alpha$-diagonal matrix $X\succ 0$ such that the matrix $\cM^\alpha(A^\ast X + X A)$ is Hurwitz. The matrix $\cM^\alpha(A^\ast X + X A)$ is Hurwitz if and only if there exist $X_1\succ 0$, $X_2\succ 0$, $\gamma$ such that
	\begin{gather*}
	B^\ast X_1 + X_1 B \prec 0, \,\,\,	\underline{\sigma}(B^\ast X_1 + X_1 B) >  \overline{\sigma}(X_1 + X_2) \gamma \delta, \\
	B^\ast X_2 + X_2 B \prec 0, \,\,\,	\underline{\sigma}(B^\ast X_2 + X_2 B)\gamma > \overline{\sigma}(X_1 + X_2) \delta, 
	\end{gather*}
	which implies that 
	\begin{gather*}
	B^\ast X_1 + X_1 B +  \overline{\sigma}(X_1 + X_2) I \delta \gamma \prec 0,  \\
	B^\ast X_2 + X_2 B +  \overline{\sigma}(X_1 + X_2) I \delta/\gamma \prec 0. 
	\end{gather*}
	It can be verified that $P = \overline{\sigma}(X_1 + X_2) X_0 (\gamma +1/\gamma)\delta$ satisfies 
	\begin{gather*}
	B^\ast P + P B +  \overline{\sigma}(X_1 + X_2) I (\gamma +1/\gamma)\delta= 0,
	\end{gather*}
    and $X_1 +X_2 \succ P$. It follows that 
	\begin{gather*}
	\overline{\sigma}(X_1 + X_2) > \overline{\sigma}(X_0) \overline{\sigma}(X_1 + X_2) (\gamma + 1/\gamma )\delta ,
	\end{gather*}
	which cannot be fulfilled since $\overline{\sigma}(X_0)  (\gamma + 1/\gamma )\delta \ge 1$ for all $\gamma>0$ due to b). Indeed,
	\begin{gather*}
	\gamma^2 - \gamma/(\delta \overline{\sigma}(X_0)) + 1\ge 0 \Leftrightarrow\\
	(\gamma - 1/(2 \delta \overline{\sigma}(X_0)))^2 + 1 - 1/(2 \delta \overline{\sigma}(X_0))^2 \ge 0 \Leftarrow\\
	1 - 1/(2 \delta \overline{\sigma}(X_0))^2 \ge 0 \Leftrightarrow  \delta \overline{\sigma}(X_0) \ge 1/2.
	\end{gather*}
	This completes the proof.
	\end{proof-of}
	\begin{rem}
		The proof of Proposition~\ref{prop:h-counter} holds when $B \in\R^{k\times k}$ with $k \ge 2$. Indeed, if $b$ is a positive scalar then $b = \underline{\sigma}(b)$, $X_0 = -1/(2 b)$ therefore we need to pick
		$\delta \ge b$. However, 
		$A =  \begin{pmatrix}
		-b & \delta \\ 
		\delta  & -b
		\end{pmatrix}$ has a positive eigenvalue if $\delta > b$ and has an eigenvalue at the origin if $\delta = b$. Hence, no two by two matrix can satisfy the conditions in the proof of Proposition~\ref{prop:h-counter}.
	\end{rem}
\end{document}